\documentclass[a4paper]{article}
\usepackage[pdftex]{graphicx}
\usepackage{amsmath}
\usepackage{amssymb}
\usepackage{amsthm}
\usepackage[utf8]{inputenc}
\usepackage[T1]{fontenc}
\usepackage{hyperref}
\usepackage{smartref}
\usepackage{todonotes}
\usepackage{fullpage}
\usepackage{epstopdf}
\usepackage{subcaption}
\usepackage[inline]{enumitem}

\newtheorem{theorem}{Theorem}
\newtheorem{definition}{Definition}

\newtheorem{remark}{Remark}

\usepackage[toc,page]{appendix}

\usepackage{color}

\renewcommand\vec{\boldsymbol}

\newcommand{\norm}[1]{\left\lVert#1\right\rVert}
\newcommand{\scal}[2]{\langle #1,#2\rangle}

\newcount\colveccount
\newcommand*\colvec[1]{
        \global\colveccount#1
        \begin{pmatrix}
        \colvecnext
}
\def\colvecnext#1{
        #1
        \global\advance\colveccount-1
        \ifnum\colveccount>0
                \\
                \expandafter\colvecnext
        \else
                \end{pmatrix}
        \fi
}

\begin{document}

\title{Structural parts as quadrics - elasticity ellipses revisited}
\author{Tam\'as Baranyai}
\maketitle

\abstract
Elasticity ellipses or central ellipses have been long used in graphic statics to capture the elastic behaviour of structural elements.
The paper gives a generalisation the concept both in dimensions and in the possibility of degenerate conics / quadrics.
The effect of projective transformations of these quadrics is also given, such that the entire mechanical system can be transformed preserving equilibrium and compatibility between its elements.

\section{Introduction}

The idea of graphically representing the elastic behaviour of a structural element can be found in many classical works of graphic statics. Possibly the best known examples are the ellipse of inertia used for graphically constructing the core (Kern) of a cross-section, and the ellipse of elasticity used to graphically construct the force in a rod, given the centre of relative rotation of its ends. 
In "Die Graphische Statik" of Culman \cite{culmann1875graphische}  one finds these as central ellipses, along with a central ellipsoid in 3D containing these ellipses. While in these cases the area moments of inertia are used, the ellipsoidal representation of the mass moments of inertia is even older, introduced by Poinsot \cite{poinsot1834} and further investigated by by Clebsch \cite{clebsch1860}.

This 3 dimensional treatment seems to be missing from later interpretations graphic statics \cite{anwendungen_egyben}, as the focus was on planar ruler and compass constructions. Eventually the use of these ellipses became sparce even in the planar case, as the graphic analysis of indeterminate structures evolved into the fixed-point method of Suter \cite{suter1923} that later took the algebraic form of the Cross-method \cite{cross1932analysis}.

Although not as wide-spread these graphic tools are still being used today, for instance in seismic analysis, \cite{papparoni,Faggella2014,FAGGELLA2017128}, the examination of historic structures \cite{ferrari2008theory} or even dental protheses \cite{Williams1991}.  

Furthermore, graphic statics is currently undergoing a renaissance partly due to computerization as it allows efficient creation of constructions visually representating the forces inside a structure \cite{hablicsek2019algebraic}. Apart from this it is used for structural optimization through application of reciprocal diagrams \cite{beghini2014structural} or through projective transformations \cite{fivet2016projective}.

This paper revisits elasticity ellipses in a more contemporary way. From the engineering standpoint it will follow the logic of numerical methods, building up a structure from a set of members with different supporting conditions represented by different stiffness matrices.  Linear members are examined first, then (sub-)structures as their sums. It will be shown how in a global coordinate system the known stiffness matrices can readily be interpreted as conics and quadrics similar to the elasticity ellipses, representing the geometrical relations the stiffness of the members entail. To this end numerous concepts from projective geometry are required, as this subject is typically not part of engineering curricula a brief description of these concepts is provided below.
 
\section{Notation, preliminaries}
Due to the mechanical motivation, we will present concepts for real (finite dimensional) projective spaces. The reader may find details in \cite{richter2011perspectives} or \cite{pottmann2001computational}.

\subsection{Projective space associated to a vector-space}
Consider $\mathbb{R}^{n+1}$ with the equivalence relation 
\begin{align}
\vec{u} \sim \vec{v} \iff \vec{u}=\lambda \vec{v} \quad \vec{u},\vec{v}\in \mathbb{R}^{n+1}, \lambda \in \mathbb{R}\setminus \{0 \}
\end{align}
(we will treat all vectors as column vectors if the distinction is necessary). Factorizing $\mathbb{R}^{n+1}\setminus \{0\}$ with this relation leads to equivalence classes $\vec{u}_{\sim}$ that can be considered points of an $n$ dimensional real projective space $PG(n)$. We will use the fact that $n-1$ dimensional projective subspaces (hyperplanes) can be represented similarly through the scalar product: point $\vec{p}_{\sim}$ lies in hyperplane $\vec{h}_{\sim}$ if and only if $\scal{\vec{p}}{\vec{h}}=0$ holds (it can be seen that the choice of vectors from the equivalence class is irrelevant.) In general a $k$ dimensional projective subspace in $PG(n)$ is identified with a $k+1$ dimensional subspace minus the origin of $\mathbb{R}^{n+1}$.

\subsection{Collineations and correlations}
We will be looking at two types of transformations: collineations and correlations. In both cases there is a bijection between all such transformations and all invertible matrix equivalence classes. As such the algebraic description on homogeneous coordinates can be given with matrix multiplication. (In the descriptions below points are $0$ dimensional projective subspaces.) 

Collineations map $k$ dimensional subspaces of $PG(n)$ to $k$ dimensional subspaces of $PG(n)$ such that all incidences are preserved. Given matrix $P$, the transformation corresponding to it's equivalence class be described with 
\begin{align}
\text{ points to points } \quad &\vec{p}_{\sim} \mapsto \vec{P}\vec{p}_{\sim} \label{eq:trafo1} \\
\text{ hyperplanes to hyperplanes } \quad &\vec{h}_{\sim} \mapsto \vec{P}^{-T} \vec{h}_{\sim}.\label{eq:trafo2}
\end{align} 

Correlations map $k$ dimensional subspaces of $PG(n)$ to $n-k-1$ dimensional subspaces of $PG(n)$ such that all incidences are preserved (point $p$ incident with hyperplane $h$ is mapped into hyperplane $p'$ incident with point $h'$). Given matrix $\vec{P}$, the transformation corresponding to it's equivalence class be described with 
\begin{align}
\text{ points to hyperplanes } \quad &\vec{p}_{\sim} \mapsto \vec{P}\vec{p}_{\sim}  \\
\text{ hyperplanes to points } \quad &\vec{h}_{\sim} \mapsto \vec{P}^{-T} \vec{h}_{\sim}. 
\end{align}
Correlations of period 2 (where $\vec{p} \sim \vec{P}^{-T}\vec{Pp} $) are called polarities, and it can be seen they correspond to invertible symmetric and anti-metric matrices. Points that are mapped to hyperplanes incident with them are called self conjugate and satisfy $\scal{\vec{p}}{\vec{Pp}}=0$. For anti-metric matrices (null-polarities) all points are such, while for symmetric matrices if they exist they correspond to conics and quadrics, as follows.

\subsection{Conics and quadrics}
The classical mechanical subject at hand seems to require the more general approach to conics (and quadrics), which goes beyond ellipses parabolas and hyperbolas, even in the plane. The reader may find a detailed introduction to the planar case in \cite{richter2011perspectives}. With the help of bilinear forms we can have a bijection between conics (2D) and quadrics (higher dimensions) and equivalence classes of symmetric matrices. The conic $\mathcal{C}$ as a set of points is given as
\begin{align}
\mathcal{C}:=\{\vec{p}_{\sim} \ \vert \ \scal{\vec{p}}{\vec{Cp}}=0 \} \label{eq:bilin}
\end{align} (we prescribe $\vec{C}=\vec{C}^T$). The mechanical interpretation presented in Section \ref{sec:member} allows the zero matrix, to which the entire projective space corresponds. This has the benefit that we can use the vector-space nature of symmetric matrices.

Given collineation acting on points with matrix $\vec{P}$, conics transform as 
\begin{align}
\vec{C} \mapsto \vec{P}^{-T}\vec{CP}^{-1} \label{eq:trafo3}
\end{align}
and conics that can be transformed into each other are called projectively equivalent. In this setting not all conics are projectively equivalent, for instance to a positive definite $\vec{C}$  the bilinear form in \eqref{eq:bilin} has no real solutions and the empty set can not be projectively mapped into a circle. We will call such "invisible" conics complex conics and we will see how the energy principles of classical mechanics often lead to them. Another way to non-equivalent conics is if $\vec{C}$ degenerates, the different number of zero eigenvalues of $\vec{C}$ will also have a mechanical interpretation.

\subsection{Lines in 3D}
Beyond points and hyperplanes we will need the Pl\"ucker coordinates of lines in $PG(3)$. They are homogeneous coordinates, we will have equivalence classes of sextuples $(l_1 \dots l_6)_{\sim}$ that represent a line if and only if they satisfy
\begin{align}
l_1l_4+l_2l_5+l_3l_6=0. \label{eq:plücker}
\end{align}
As such equivalence classes can be considered points in $PG(5)$, we can imagine all lines as a subset in $PG(5)$. It can be seen that \eqref{eq:plücker} is in fact an equation of a quadric, meaning the set of all lines forms a quadric in $PG(5)$ called \emph{Klein quadric} (we will denote it with $\mathcal{Q}$). One can compute the effect of 3 dimensional correlations and collineations on Pl\"ucker coordinates, leading to a linear map $\mathbb{R}^6 \rightarrow \mathbb{R}^6$, which can be thought of as a correlation or collineation in $PG(5)$. We will rely on the following theorem linking the two:
\begin{theorem}[\cite{pottmann2001computational}]\label{thm:1}
Projective collineations and correlations of $PG(3)$ induce projective automorphisms of the Klein quadric, and the Klein quadric does not admit any other projective automorphisms.
\end{theorem}

\subsection{Posing the mechanical problem as a vector space}
We will look at structures with the usual assumptions of the Euler-Bernoulli beam theory. We can pose the arising mechanical problems such that both static and kinematic dynames can be elements of a respective vector space, we can associate a projective space to. The idea started from Sir Robert Ball's Screw Theory \cite{ball1900treatise}, the reader may find more in \cite{pottmann2001computational} or \cite{davidson2004robots}.

After a choice of coordinate system the effect of any force system can be given with a force vector $\vec{F}\in\mathbb{R}^3$ and a moment vector $\vec{M}\in\mathbb{R}^3$ with respect to the origin. We can combine them into a single vector $\vec{f}=(\vec{M},\vec{F})\in \mathbb{R}^6$.

For the kinematic dyname it may be useful to start from the better known instantaneous kinematics: the velocity state of a rigid body can be described with the vector pair $(\vec{\Omega},\vec{V})$ where $\vec{\Omega}\in \mathbb{R}^3$ is the angular velocity of the body as it rotated around an axis passing through the origin and $\vec{V}\in \mathbb{R}^3$ is the translational velocity of the origin. We can get the small displacement approximation we are going to use by letting the velocity state act for a small time, displacing each point in the direction of its velocity. The effect of this can be captured in the kinematic dyname $\vec{d}=(\vec{\Phi},\vec{\Delta})\in \mathbb{R}^6$, where $\vec{\Phi}\in \mathbb{R}^3$ describes rotation about an axis passing through the origin while $\vec{\Delta}\in \mathbb{R}^3$ describes the displacement of the origin.

The effect of these dynames can be represented by a single force or rotation if and only if the sextuple satisfies \eqref{eq:plücker}. In these cases we can think of these mechanical quantities as line representants, and the line they represent is the line of action of the force, or the axis of rotation. (In both cases it might be an ideal line at infinity, corresponding to moments in the static and translations in the kinematic case.)\\

Displacement $\vec{d}$ will be treated as a point in $PG(5)$, given by the equivalence class $\vec{d}_\sim$. This is nothing else then the Klein-embedding of lines into 5 dimensions, extended to kinematic dynames not reducible to a single rotation.

Static dyname $\vec{f}$ will be identified with a 4 dimensional hyperplane of $PG(5)$ given by equivalence class $\vec{f}_\sim$. This is dual to the usual Klein embedding, dyname $\vec{f}$ is reducible to a single force if the corresponding hyperplane is a tangent hyperplane of the Klein quadric.

In what follows we will give the relation of the mechanical properties and will treat them directly as points or hyperplanes with the equivalence signs neglected.

\section{Structural parts as conics}
In numerical analysis it is typical to decompose complex structures into pieces with known behaviours. In case of frames these known elements are usually linear, rods having a defined axis and connecting joints or vertices of the structure. To each known element or in certain cases sets of elements corresponds a stiffness matrix. Here we will show how these matrices can be considered conics and how these conics contain geometrical information relevant to the forces and displacements involved. The stiffness matrix $\vec{K}$ connects displacements and forces as 
\begin{align}
\vec{f}=\vec{Kd}.
\end{align}
Since displacements are identified with points of $PG(5)$ while forces with $4$ dimensional subspaces, if $\vec{K}$ is invertible we will interpret this as a correlation and denote it with $\kappa$. In case of a degenerate matrix the map from points to hyperplanes is meaningful, but the entire incidence structure of $PG(5)$ is not preserved. The set
\begin{align}
\mathcal{K}:=\{\vec{d} \ \vert \ \scal{\vec{d}}{\vec{Kd}}=0 \}
\end{align}
will be called the corresponding \emph{stiffness conic} or \emph{stiffness quadric}. In contrast the name \emph{elasticity ellipse / ellipsoid / quadric} will be used for shapes corresponding to a correlation different then $\kappa$. When using the elasticity ellipse the difference is corrected by adding a geometric operation (typically mirroring) before or after the correlation given by the elasticity quadric.

We will examine a structural member first, then describe the effect of different support conditions and how structures can be built from members. This will be followed by visualization methods and a few results on projective transformations of the mechanical systems.

\subsection{Stiffness conics of a structural member}\label{sec:member}
Consider a rod joining vertices $i$ and $j$! If vertex $j$ is displaced relatively to vertex $i$ with $\vec{d}_{i,j}$, 
force $\vec{f}_{ij}$ will act on the $j$ end of the rod and on vertex $i$ while force $\vec{f}_{ji}=-\vec{f}_{ij}$ will act on vertex $j$ and on the $i$ end of the rod. We can describe the stiffness of the rod with a stiffness matrix as
\begin{align}
\vec{f}_{ij}=\vec{Kd}_{ij}.
\end{align}
We know, that $\vec{K}=\vec{K}^T$ due to Betti's theorem \cite{roller_szabo}, implying that all it's eigenvalues are real and there is at least one eigenvector to each. Furthermore, no negative eigenvalue is possible: if $\vec{K}\vec{v}=\lambda \vec{v}$ existed with $\lambda<0$, we could consider end $i$ clamped and end $j$ free for the moment and apply force $\vec{f}_{ij}=\lambda \vec{v}$ on end $j$, resulting in displacement $\vec{d}_{ij}=\vec{v}$. The own work of the force on the displacement it caused would be $\frac{1}{2}\scal{\vec{f}_{ij}}{\vec{d}_{ij}}=\frac{1}{2}\scal{\vec{v}}{\vec{K}\vec{v}}=\frac{\lambda}{2}\norm{\vec{v}}^2$, which can not be negative. This implies that equation $\scal{\vec{d}_{ij}}{\vec{K} \vec{d}_{ij}}=0$ is either never satisfied in a real vector-space ($\vec{K}$ is positive definite) or all solutions are inside the kernel of $\vec{K}$ ($\text{ker}(\vec{K})$). Geometrically speaking $\mathcal{K}$ is either a complex conic not appearing in real projective space, or a degenerate conic corresponding to a projective subspace. 

The number of zero eigenvalues and the dimensionality of $\text{ker}(\vec{K})$ depends on the supporting conditions on the ends of the rod. On the displacement side $\mathcal{K}$ is precisely the set of displacements that can happen with no arising forces. On the side of forces any non-zero force $\vec{f}_{ij}$ must be in the image space of $\vec{K}$ ($\text{im}(\vec{K})$). Since for symmetric matrices  $\text{im}(\vec{K})=\text{ker}(\vec{K})^\perp$ holds, we have
\begin{align}
\vec{f}_{ij}\neq 0 \implies \vec{f}_{ij}^T\vec{d}=0 \ \forall \vec{d}\in \mathcal{K},
\end{align}  
that is any force arising from relative displacement $\vec{d}_{ij}$ must be incident with \emph{all} points of the stiffness conic. It is not hard to see, that as the supports become less strict $\text{ker}(\vec{K})$ grows in dimension, less types of forces are possible leading to stricter incidence conditions given by $\mathcal{K}$. This is more pronounced in the case of planar problems, where degenerate conics directly appear in relation to the geometry of the structures. A few examples illustrating this are presented in Appendix \ref{app:ex} and \ref{app:duex}.

\subsection{Combined effect of members}
One use of tying our structural elements to coordinate-free matrices forming a vector-space is that we may use their linear combinations, representing the combined effect of these members. The idea of an elasticity ellipse of a set of elements is not new, we may find it in the works of Culmann and Richter\cite{anwendungen_egyben}. They give elasticity ellipses for cells of trusses (Fach) considering different geometries. One could generalize their method of testing the structure to appropriately chosen displacements, but we are in a better position thanks to the linear algebraical treatment of conics.\\

Consider two rigid bodies $a$ and $b$, connected by a set of elements numbered $i \in \{1 \dots n\}$, with corresponding stiffness conics $\vec{K}_i$. Given relative displacement $\vec{d}_{ab}$ the force from the displacement in each element is $\vec{f}_i=\vec{K}_i\vec{d}_{ab}$ (acting on body $a$). As the total force acting on body $a$ is 
\begin{align}
\sum_i \vec{f}_i=\sum_i \left(\vec{K}_id_{ab}\right)=\left(\sum_i \vec{K}_i\right)\vec{d}_{ab}
\end{align}   
we have deduced that the stiffness conic of the combined elements is the sum of the stiffness conics of the parts.\\

All the things stated for the stiffness conics of members can be stated for stiffness conics of their sums. The sum of positive (semi-) definite matrices will be positive (semi-) definite and the types of conics corresponding to sums of parts will be the same as in the case of the members. Adding more members to a structure will decrease the dimensionality of $\text{ker}(\vec{K})$ and thus $\mathcal{K}$, implying a looser incidence condition on the forces.

\subsection{Visualization of the non-degenerate case}
In a lot of cases we build structures that resist all types of motion and their stiffness conics are complex - invisible in real projective space. A way around this is given by the idea of elasticity ellipses \cite{culmann1875graphische}, giving a graphical way to construct points of displacements and lines of forces from each other. We will extend this idea to the case of spatial forces and displacements, resulting in a 5 dimensional elasticity quadric.  We will again consider the case of a single member to have a concrete example, the arguments except for the 3D visualization part generalize as provided in the previous subsection.\\

\subsubsection{The 5 dimensional elasticity quadric}
Consider a rod of length $L$ joining vertices $i$ and $j$! Let us pick the coordinate system such that the origin is in the midpoint of the rod, let the rod be parallel with the $x$ axis and let $y$ and $z$ be the principial directions of its cross-section. We will denote the area of the cross section with $A$, the principial inertia moments with $I_y$ and $I_z$ and the polar inertia of the cross-section with $I_x$. The Young-modulus will be denoted with $E$, the shear modulus with $G$. It is given in \cite{Livesley}, how in this coordinate system the mechanical behaviour of a "straight, uniform" member gives the map:
\begin{align}
\vec{f}_{ij}=\text{diag}\left(GI_x/L,EI_y/L,EI_z/L,EA/L,12EI_z/L^3,12EI_y/L^3\right)\vec{d}_{ij}\label{eq:mechmap}
\end{align}
where diag(  ) is shorthand for diagonal matrix. In engineering books there are a number of tacit or explicit assumptions involved, when "straight, uniform" or "homogeneous" members are used. To avoid confusion a formal definition is presented what the paper will mean under straight uniform rods:

\begin{definition}[straight uniform rod]\label{def:1}
A line segment with cross-sections as fictitious rigid objects corresponding to each point on the line segment, connected with neighbouring cross sections elastically. In the stress-free case the location of the centroids of the cross sections as well as their sizes and orientations have to change in a continuous way, when considered as functions over the line segment. The relative motion of cross sections at one endpoint and an internal point of the axis (the line of the line-segment)
\begin{enumerate}[label=\roman*)]
\item under axial force is  pure translation in the axial direction with magnitude proportional to the length of line segment between the points.
\item under torsion is rotation around the axis with magnitude proportional to the length of the line segment between the points.
\item under pure bending moment is rotation around an axis the direction of which is independent of the length of the line segment between the points. The magnitude of the rotation is proportional to the length of the line segment between the points. 
\end{enumerate}
\end{definition}

\begin{remark}
Shear forces and deformations perpendicular to the axis of the rod are missing as the Euler-Bernoulli beam theory neglects shear deformations. The effect of shear forces is captured in the fact they cause bending moments. The displacements orthogonal to the axis arise from the relative rotations of the cross-sections of the rod.
\end{remark}

\begin{remark}
This definition is stricter than the rod having the same cross-section everywhere, as the centroid and shear-center of the cross-sections have to coincide.
\end{remark}

In order to visualise the mechanical behaviour given in \eqref{eq:mechmap}, let us introduce the collineation $\tau$ represented by matrix $\vec{T}$ and the correlation $\underline{\kappa}$ represented by $\underline{\vec{K}}$, where:
\begin{align}
\vec{T}:=&\text{diag}(-1,-1,-1,1,1,1)\\
\underline{\vec{K}}:=&\vec{KT}.
\end{align}
We can see, that $\vec{T}=\vec{T}^{-T}$, and $\vec{K}=\vec{T}\underline{\vec{K}}=\underline{\vec{K}}\vec{T}$, meaning we have a construction similar to the elasticity ellipse of Culmann. As such, we will call the set
\begin{align}
\mathcal{\underline{K}}:=\{\vec{d} \ \vert \ \scal{\vec{d}}{\underline{\vec{K}}\vec{d}}=0 \}\label{eq:equadric}
\end{align}
the elasticity quadric of the member.

\begin{remark}
This is not the only way to visualize $\vec{K}$, even in the planar case if we mirrored with respect to a line we would have an elasticity hyperbola and not an ellipse. The map $\tau$ has been selected because it seems the most consistent with earlier works of Culmann and Ritter, while having the property that the Klein quadric is invariant under it. In fact it can be interpreted as an action on 3 dimensional lines, mirroring them with respect to the centre of the coordinate system.
\end{remark}

\subsubsection{Notable sections of the elasticity quadric}
Recall the radii of inertia being defined as
\begin{align}
i_y:=\sqrt{\frac{I_y}{A}} \ \text{and} \ i_z:=\sqrt{\frac{I_z}{A}}.
\end{align} 
Let us restrict ourself to the $\Phi_x=\Phi_y=\Delta_z=0$ subspace (which is $PG(2)$) and adopt the drawing convention that the Euclidean (finite) points are represented with vectors satisfying $\Phi_z=1$. (Other displacements that are a scalar multiple of this appear on the same place on the projective plane.) We can use the radii of inertia and multiply \eqref{eq:equadric} with $\frac{L}{EA i_z^2}$ giving an equivalent equation of $\mathcal{\underline{K}}$ (restricted to this subspace) as
\begin{align}
-\Phi_z^2+\frac{\Delta_x^2}{i_z^2}+12\frac{\Delta_y^2}{L^2}=0
\end{align}
which is nothing else then the elasticity ellipse of classical planar graphic statics. A similar observation can be made in the $\Phi_x=\Phi_z=\Delta_y=0$ subspace. The mechanical problem appears rotated with $\pi/4$ in these planes, due to the way we represent cross products with scalar products.\\

Furthermore, in the  $\Phi_x=\Delta_y=\Delta_z=0$ subspace drawn with finite points corresponding to $\Delta_x=1$ the section is a dual ellipse of the usual ellipse used for graphically constructing the core of a cross-section. This can be seen by multiplying \eqref{eq:equadric} with $\frac{L}{EA}$ giving 
\begin{align}
-\Phi_y^2 i_y^2-\Phi_z^2 i_z^2+1=0.
\end{align} 
The classic graphic construction connects points of attack of forces to the neutral axes, where the stresses are zero. It is not hard to see if we consider an infinitely short rod the relative motion of the endpoints will describe the relative motion of two neighbouring rigid plates of the rod model and the neutral axis of stresses is the axis of rotation. In this setting forces are identified with points and displacements with lines, which is the dual of our setting, hence the dual ellipse (see \cite{richter2011perspectives}) in the  $\Phi_x=\Delta_y=\Delta_z=0$ subspace. The primal ellipse of Culmann in this case would correspond to $\vec{K}^{-1}$, and the coordinate system is again rotated with $\pi/4$.

Another type of problems where we may only consider a 3 dimensional subspace is the behaviour of grillages. In Appendix \ref{app:duex} a few examples are provided, with elasticity ellipses again appearing as dual ellipses.

\subsubsection{Three dimensional graphic representation}
As we typically see forces and rotations having lines of action in $PG(3)$ and would like to interpret the mechanical behaviour in 3 dimensions instead of 5, two natural questions arise:
\begin{enumerate}[label=\roman*)]
\item Can we represent the behaviour of the member with a 3 dimensional quadric?
\item Is it true that kinematic dynames reducible to a rotation around an axis in $PG(3)$ are mapped to static dyname reducible to a force having a line of action in $PG(3)$?
\end{enumerate}

We know from Theorem \ref{thm:1} that geometrically speaking these two questions are equivalent if we consider all possible lines with corresponding static and kinematic dynames. We will show that there are rods for which the answer to these two questions is unconditionally "yes". We will also show that for all rods the answer "it depends" is also applicable, implying a condition on the geometry of lines involved.\\

Let us embed the Euclidean space into $PG(3)$ such that finite points have a representant of shape $(x,y,z,1)$, while ideal points of shape $(x,y,z,0)$! To each rod we can create a polarity $\kappa_3$ represented by 
\begin{align}
\vec{K}_3:=\text{diag}\left(\frac{12}{L^2},i_z^{-2},i_y^{-2},1\right)
\end{align}
acting on homogeneous coordinates of $PG(3)$. The corresponding 3 dimensional quadric $\mathcal{K}_3$ is again complex. We can do the same as we did in 5 dimensions and introduce $\tau_3$ to be a mirroring with respect to the centre of the coordinate system, and another polarity defined as $\underline{\kappa}_3:=\kappa_3 \circ \tau_3$ (mirroring first, but in this particular case the order is irrelevant). Note, how the effect of $\tau_3$ on lines is the restriction of $\tau$ to the Klein quadric.
 
It is easy to see that the self conjugate points of $\underline{\kappa}_3$ form an ellipsoid $\underline{\mathcal{K}}_3$ with equation 
\begin{align}
-\frac{12}{L^2}x^2 -\frac{y^2}{ i_z^2}-\frac{z^2}{ i_y^2}+1=0
\end{align} 
which can be used to visualize and graphically construct the effect of $\kappa_3$. With this, we can more formally give answers to the aforementioned questions, as:

\begin{theorem}\label{thm:repi}
For static and kinematic dynames reducible to lines correlation $\kappa_3$ gives correct lines of action 
\begin{enumerate}[label=\roman*)]
\item for all lines if and only if the mechanical properties of the rod satisfy $GI_x=12\frac{EAi_y^2i_z^2}{L^2}$
\item for all rods, if and only if the dynames satisfy $\scal{\vec{d}}{\vec{l}_1}=0 \iff \scal{\vec{f}}{\vec{l}_1}=0$ or\\ $\vec{d}=\lambda_d \vec{l}_1 \iff \vec{f}=\lambda_f \vec{l}_1$ with $\vec{l}_1=(1,0,0,0,0,0)$ and some $\lambda_d,\lambda_f \in \mathbb{R}\setminus \{ 0 \}$.
\end{enumerate}
\end{theorem}

\begin{proof}
In order to show $i)$ we have to show that the correlation induced by $\kappa_3$ in $PG(5)$ is $\kappa$, that is the linear maps describing the two 5 dimensional correlations are scalar multiples of each other. Given points $(p_x,p_y,p_z,1)=(\vec{p},1)$ and $(q_x,q_y,q_z,1)=(\vec{q},1)$, Pl\"ucker coordinates of lines passing through them can be calculated (see \cite{pottmann2001computational}) as 
\begin{align}
(1\vec{p}-1\vec{q},\vec{q}\times\vec{p}).
\end{align} 
Polarity $\kappa_3$ maps points $(\vec{p},1)$ and $(\vec{q},1)$ into planes $\text{diag}(12/L^2,i_z^{-2},i_y^{-2})(\vec{p},1)$ and $\text{diag}(12/L^2,i_z^{-2},i_y^{-2})(\vec{q},1)$.
The intersection line of these planes has Pl\"ucker coordinates
\begin{align}
\left(\text{diag}(12/L^2,i_z^{-2},i_y^{-2})\vec{q}\times \text{diag}(12/L^2,i_z^{-2},i_y^{-2})\vec{p},\text{diag}(12/L^2,i_z^{-2},i_y^{-2})(1\vec{p}-1\vec{q})\right). \label{eq:plsikmetszet}
\end{align}
Using that for any invertible $3\times 3$ matrix $\vec{A}$ and vectors $\vec{a}$ and $\vec{b}$
\begin{align}
(\vec{A}\vec{a})\times(\vec{A}\vec{b})=\text{det}(\vec{A})\vec{A}^{-T}(\vec{a}\times\vec{b})
\end{align}
holds the expression in \eqref{eq:plsikmetszet} turns into:
\begin{align}
\left(\text{diag}(i_y^{-2}i_z^{-2},12i_y^{-2}/L^2,12i_z^{-2}/L^2)(\vec{q}\times\vec{p}),\text{diag}(12/L^2,i_z^{-2},i_y^{-2})(1\vec{p}-1\vec{q})\right).
\end{align}
As such, the correlation induced by $\kappa_3$ in $PG(5)$ can be described by the linear map
\begin{align}
\kappa_g: \vec{d} \mapsto \vec{f}=\text{diag}\left(12/L^2,i_z^{-2},i_y^{-2},i_y^{-2}i_z^{-2},12i_y^{-2}/L^2,12i_z^{-2}/L^2\right)\vec{d}.\label{eq:graphmap}
\end{align}
Comparing this diagonal matrix to the one in expression \eqref{eq:mechmap}, we see that they are scalar multiples of each other if and only if
\begin{align}
\frac{GI_x}{L}&=12\frac{EAi_y^2i_z^2}{L^3}\label{eq:feltetel}
\end{align} holds.

The validity of $ii)$ can also be seen by comparing equations \eqref{eq:mechmap} and \eqref{eq:graphmap}.
\end{proof}

The 3 dimensional elasticity ellipsoid does directly contain both ellipses of classical graphic statics, and the planar constructions can be interpreted as the intersection of spatial constructions and an image plane, as illustrated in Figure \ref{fig:spatial}. The 3 dimensional ellipsoid itself was known in the late 1800's \cite{culmann1875graphische} but the spatial use of it does not seems widespread. This may be due to the technical restrictions of that time, in our days computers will solve the required operations on homogeneous coordinates in fractions of a second.\\

We can also interpret Theorem \ref{thm:repi} in 5 dimensions, relating the elasticity quadric $\underline{\mathcal{K}}$ to the Klein-quadric $\mathcal{Q}$. According to the $i)$ part and the fact that $\mathcal{Q}$ is invariant under $\tau$ iff condition $GI_x=12\frac{EAi_y^2i_z^2}{L^2}$ is satisfied, $\underline{\mathcal{K}}$ is such that its tangent hyperplanes  at the intersection $\underline{\mathcal{K}} \cap \mathcal{Q}$ also touch $\mathcal{Q}$, although not necessary at the same point. The set $\mathcal{L}:=\{ \vec{d} \ \vert \ \vec{d} \in \mathcal{Q} \text{ and }  \scal{\vec{d}}{\vec{l}_1}=0 \}$ is a special linear complex, the intersection of the Klein-quadric with one of its tangent hyperplanes. According to the $ii)$ part regardless of the torsion stiffness the tangent hyperplanes of $\underline{\mathcal{K}}$ at the intersection $\underline{\mathcal{K}} \cap \mathcal{L}$ also touch $\mathcal{Q}$. (Point $\vec{l}_1$ is not on $\underline{\mathcal{K}}$.)

\begin{figure}[h]
\centering
\includegraphics[width=0.5\textwidth]{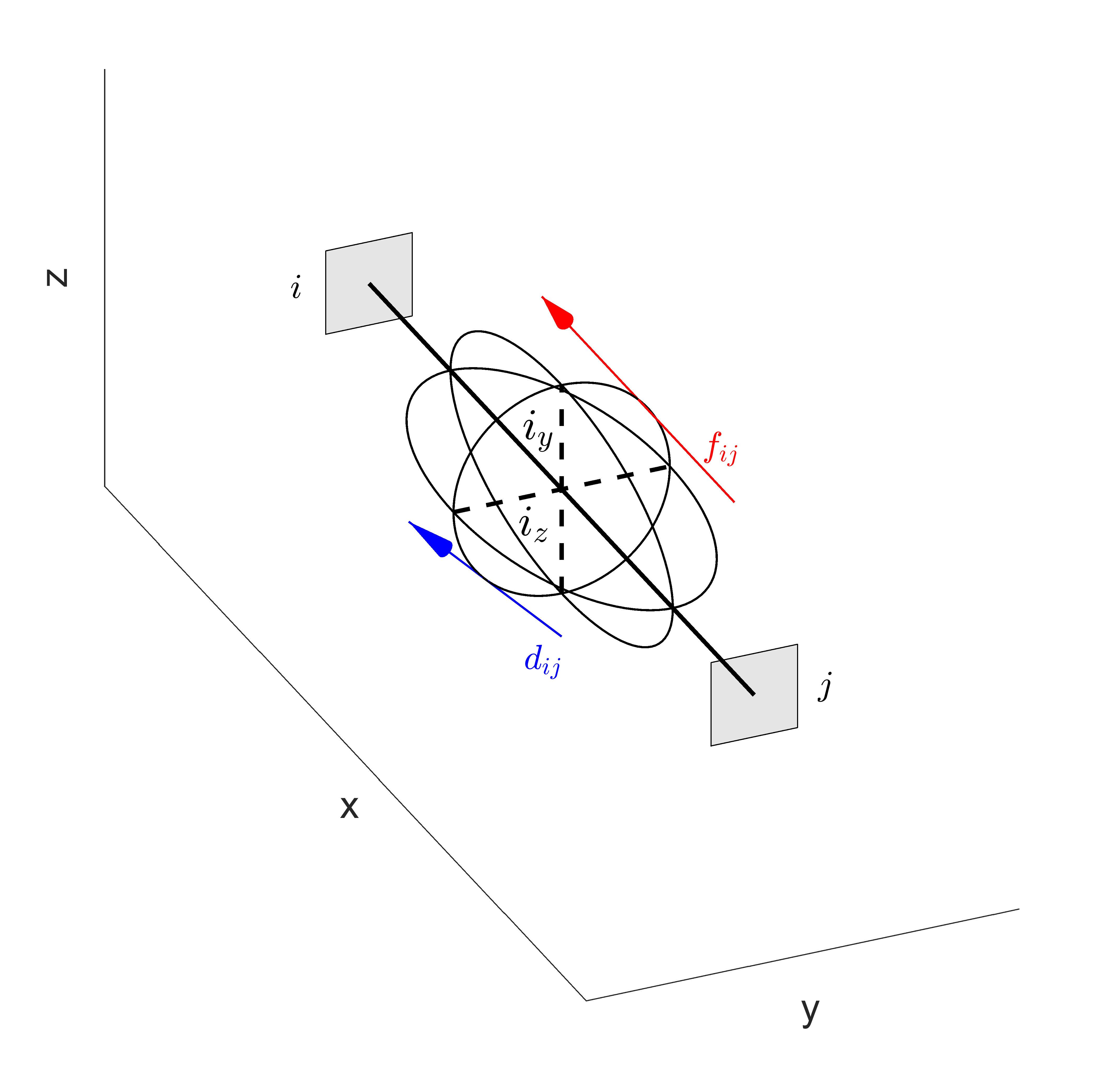}\includegraphics[width=0.5\textwidth]{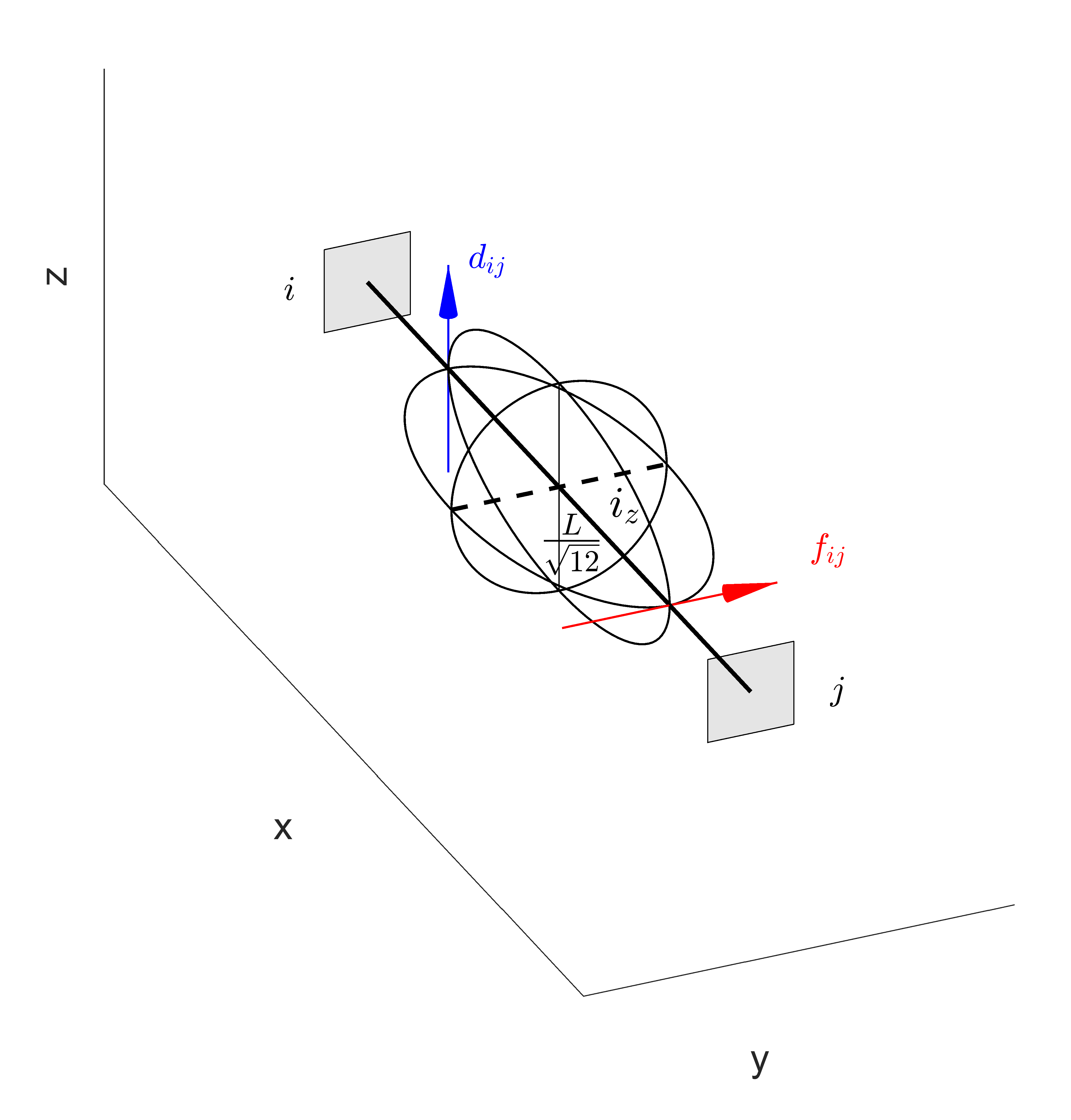}
\caption{Spatial interpretation of previous planar ellipses of graphic statics. Left: choosing the $x=0$ plane as an image plane the ellipse of the cross section appears connecting (intersection) points of attack of forces to (intersection) lines of attack of rotations -the neutral axes in case of cross-sections. Right: choosing the $y=0$ plane as image plane the elasticity ellipse of planar graphic statics appears, connecting (intersection) points of rotation to (intersection) lines of action of forces}\label{fig:spatial}
\end{figure} 

\subsection{Transformations}
As our mechanical properties are tied to the geometry we can take the projective transform of an entire mechanical system, with linear maps \eqref{eq:trafo1},\eqref{eq:trafo2} and\eqref{eq:trafo3}. In practice we are interested in 3 or 2 dimensional collineations but they can be analysed through the 5 dimensional collineations they induce. While these relations define a matrix equivalence class $\vec{P}_{\sim}$ for collineation $\pi$, we have additional requirements that will narrow down the possibilities. Yet, we will see that there is no unique way to transform the mechanical system corresponding to a projective change in geometry.  To show this let us chose a fixed representant $\vec{P}\in \vec{P}_{\sim}$ to describe the effect of $\pi$ on the mechanical system. The system preserves compatibility of displacements with each other if for all $\vec{d}_1$ and $\vec{d}_2$
\begin{align}
\pi(\vec{d}_1+\vec{d}_2)=\pi(\vec{d}_1)+\pi(\vec{d}_2)
\end{align} 
holds, implying 
\begin{align}
\pi(\vec{d})=\lambda_d \vec{Pd} \quad \text{ for a fixed }\lambda_d \neq 0.\label{eq:trd}
\end{align} 
Similarly, in order to preserve static equilibrium, the transformations need to satisfy:
\begin{align}
\pi(\vec{f}_1+\vec{f}_2)=\pi(\vec{f}_1)+\pi(\vec{f}_2) \quad \forall \vec{f}_1,\vec{f}_2
\end{align}implying\begin{align}
\pi(\vec{f})=\lambda_f \vec{P}^{-T}\vec{f} \quad \text{ for a fixed }\lambda_f \neq 0. \label{eq:trf}
\end{align}
In order to be able to take sums of stiffness conics we need
\begin{align}
\pi(\vec{K}_1+\vec{K}_2)=\pi(\vec{K}_1)+\pi(\vec{K}_2) \quad \forall \vec{K}_1,\vec{K}_2 \end{align}implying\begin{align}
\pi(\vec{K})=\lambda_K \vec{P}^{-T}\vec{KP}^{-1} \quad \text{ for a fixed }\lambda_K \neq 0. \label{eq:trk}
\end{align}
Finally, in order to preserve the compatibility of forces and displacements we need
\begin{align}
\pi(\vec{Kd})=\pi(\vec{K})\pi(\vec{d}) \quad \forall \vec{K},\vec{d} \implies
\lambda_K=\frac{\lambda_f}{\lambda_d}.
\end{align}
The effect of $\lambda_K$ can be considered as scaling the Young moduli $E$ of the materials involved, as it is the only linear term present in all the stiffness matrices ($G$ can be expressed from $E$ using Poisson's ratio). In the end we can choose two of the three $\lambda$ values freely.\\

Whether the given transformation corresponding to collineation $\pi$ makes sense or not may depend on the problem at hand. A deep categorization of problems in this respect is left for another occasion, we only present the following theorem providing a safely usable subset of transformations. (The paper uses the usual categorization of transformations: Euclidean $\subset$ similarity $\subset$ affine $\subset$ projective.)

\begin{theorem}\label{thm:trafo}
Exactly similarity transformations preserve the straight uniform property of the rods.
\end{theorem}

\begin{proof}
We will check properties $i)$, $ii)$ and $iii)$ of Definition \ref{def:1}. In each case we will split the rod (with endpoints $i$ and $j$) in two, at point $k$ in-between. The proportionality of the magnitudes of the displacements will be checked by comparing them on the rod-segments. Let us denote the lengths of the two segments with $L_{ik}$ and $L_{kj}$, the 3 dimensional similarity transformation with $\sigma$ and the matrix describing effect on displacements with $\vec{S}$. As similarity transformations are a subset of affine transformations (preserving ratios of parallel line segments), we have 
\begin{align}
\frac{L_{ik}}{L_{kj}}=\frac{L_{\sigma (i)\sigma (k)}}{L_{\sigma (k)\sigma (j)}}
\end{align}
where ${L_{\sigma (i)\sigma (k)}}$ denotes the distance between the images of $i$ and $k$ under $\sigma$.
\begin{enumerate}[label=\roman*)]
\item Let us apply an axial (normal) force $\pm \vec{f}_A$ with magnitude $N$ on the endpoints of the rod, implying $\vec{f}_{ik}=\vec{f}_A=\vec{f}_{kj}$. The own works of force $\vec{f}_A$ on the two axial displacements $\delta_{A,ik}$ and $\delta_{A,kj}$ are
\begin{align}
\frac{1}{2}N \delta_{A,ik}=\frac{1}{2}\scal{\vec{f}_A}{\vec{K}^{-1}_{ik} \vec{f}_A}\\
\frac{1}{2}N \delta_{A,kj}=\frac{1}{2}\scal{\vec{f}_A}{\vec{K}^{-1}_{kj} \vec{f}_A}
\end{align}
such that 
\begin{align}
\frac{L_{ik}}{L_{kj}}=\frac{\delta_{A,ik}}{\delta_{A,kj}}=\frac{\scal{\vec{f}_A}{\vec{K}^{-1}_{ik} \vec{f}_A}}{\scal{\vec{f}_A}{\vec{K}^{-1}_{kj} \vec{f}_A}}
\end{align}
holds due to the homogeneity of the rod. The proportion of the transformed axial displacements can be calculated to be
\begin{align}
\frac{\scal{\vec{S}^{-T}\vec{f}_A}{(\vec{S}^{-T}\vec{K}_{ik}\vec{S}^{-1})^{-1} \vec{S}^{-T}\vec{f}_A}}{\scal{\vec{S}^{-T}\vec{f}_A}{(\vec{S}^{-T}\vec{K}_{kj}\vec{S}^{-1})^{-1} \vec{S}^{-T}\vec{f}_A}}=\frac{\scal{\vec{f}_A}{\vec{K}^{-1}_{ik} \vec{f}_A}}{\scal{\vec{f}_A}{\vec{K}^{-1}_{kj} \vec{f}_A}}=\frac{L_{ik}}{L_{kj}}=\frac{L_{\sigma (i)\sigma (k)}}{L_{\sigma (k)\sigma (j)}} \label{eq:aranyos}
\end{align}
implying the transformed displacements are proportional to the  lengths of the transformed rod segments. 

The fact that $\vec{Sd}_{ik}$ and $\vec{Sd}_{kj}$ represent pure axial displacements can be seen through the angle preserving nature of the similarity transformations. In the 3D setting the axial translational displacements are represented by an ideal line, the intersection line of the planes orthogonal to the axis of the rod. As this orthogonality is preserved, the transformed displacement will be represented by in ideal line lying in the planes orthogonal to the transform of the axis. Note how more general affine transformations would not preserve the axial direction of the displacements.

\item If we replace $\vec{f}_A$ in $a)$ with a torsional moment we can repeat the shown calculation leading to an equation similar to \eqref{eq:aranyos}, implying the proportionality of axial rotations. The axis of the transformed rotation being the rod axis follows trivially. 

\item If we replace $\vec{f}_A$ in $a)$ with a pure bending moment $\vec{f}_M$ and repeat the calculation leading to \eqref{eq:aranyos} we get the proportionality of the rotations from bending. As the starting rod was straight uniform the starting displacements $\vec{d}_{ik}=\vec{K}^{-1}_{ik} \vec{f}_M$ and $\vec{d}_{kj}=\vec{K}^{-1}_{kj} \vec{f}_M$ have parallel lines as axes. The similarity transformation preserves this parallelism, which completes the proof. 
\end{enumerate}

\end{proof}

\section{Summary}
Structural members have been identified with conics and quadrics whose points, if they exist in real projective space express an incidence condition on the deformational forces that can arise. If no such points exist the elasticity conics and quadrics can be used as detailed. The latter exists in 3 and 5 dimensions as well, the applicability of the 3 dimensional ellipsoid depends on the properties of the member as well as the loads, as detailed in Theorem \ref{thm:repi}.

The description of projective transformations of the entire system of forces, displacements and stiffness relations has been presented, along with a basic result showing straight uniform rods stay straight uniform under similarity transformations (Theorem \ref{thm:trafo}). Notably, 
as Fivet \cite{fivet2016projective} pointed out affine transformations preserve the applicability of the gravity field while general projective collineations do not. If one wishes to preserve the straight uniform property of the rods (along with the compatibility of forces with displacements), further restriction to a subset of affine transformations is necessary. For specific problems more general transformations may be usable, the investigation of this is beyond the scope of this work.

The description provided in the paper has multiple properties that make it suitable for computerised use. The treatment of dynames as vectors and the vector-space nature of symmetric matrices makes calculation with them fast and efficient. The construction of structures and sub-structures from previously given known elements is methodical and can easily be automated.  


\section{Declarations}
The work has been supported by the \'UNKP-20-4 New National Excellence Program of the Ministry for Innovation and Technology from the source of the National Research, Development and Innovation Fund.

\appendix
\section{Planar examples}\label{app:ex}
A few planar examples are provided, where the degenerating conics can readily be seen in relation to the mechanical problem at hand. We will consider the Euclidean $x,y$ plane, and embed it into $PG(2)$ with $(x,y,1)$. If we represent forces with $(-F_y,F_x,M_z)$ as homogeneous line coordinates and displacements as $(-\Delta_y,\Delta_x,\Phi_z)$ as homogeneous point coordinates, we get correct lines and points of action. The three examples covered are presented in Figure \ref{fig:app_rudak}.\\ 

\begin{figure}[h]
\centering
\includegraphics[width=1\textwidth]{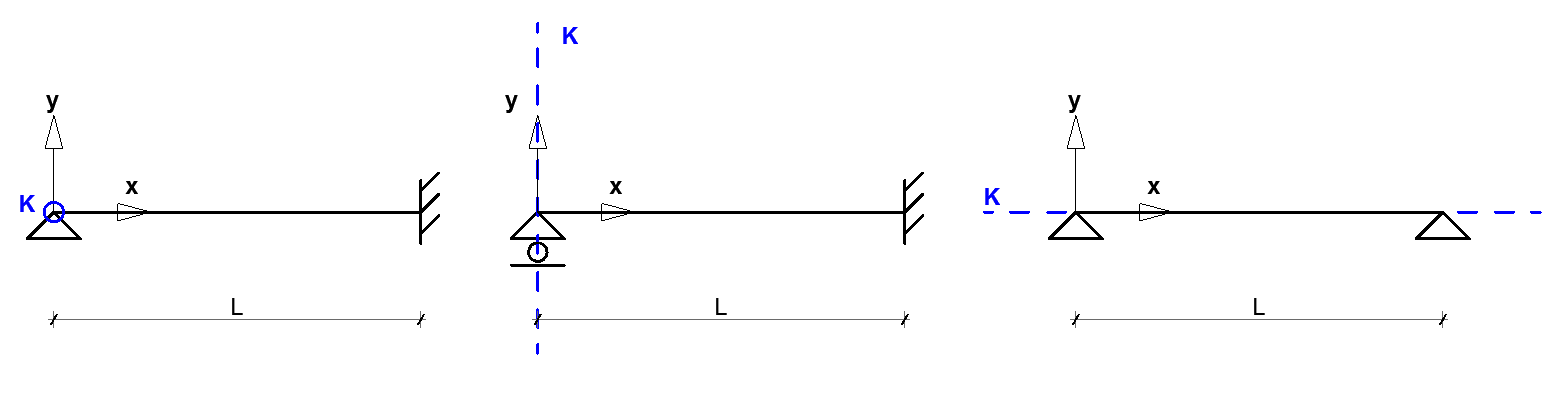}
\caption{Degenerate conics $\mathcal{K}$ appearing as incidence conditions on forces. Left: single point (the origin), where the force must pass thorugh. Middle and right: lines of action the force must have}\label{fig:app_rudak}
\end{figure}

The first case to the left has stiffness matrix
\begin{align}
\vec{K}=\begin{bmatrix}
3EI_z/L^3 & & \\
 & EA/L & \\
 &  & 0 
\end{bmatrix}
\end{align}
with the corresponding stiffness conic 
\begin{align}
3x^2/L^2+y^2/i^2_z=0,
\end{align}  satisfied by $(0,0)$, the point all forces from the kinematic load must go through.

The second case in the middle has stiffness matrix
\begin{align}
\vec{K}=\begin{bmatrix}
3EI_z/L^3 & & \\
 & 0 & \\
 &  & 0 
\end{bmatrix},
\end{align} 
which represents a conic degenerated to the line
\begin{align}
3x^2/L^3=0 \rightarrow x=0
\end{align}
which is the line of action of the arising force from kinematic loads.

The third case to the right has stiffness matrix 
\begin{align}
\vec{K}=\begin{bmatrix}
0 & & \\
 & EA/L & \\
 &  & 0 
\end{bmatrix}
\end{align}
which represents line $y=0$, again the line of action of the arising force from kinematic loads.

\section{Dual examples}\label{app:duex}

We will consider elements of a grillage in the $x,y$ plane, with a 3 dimensional subset of forces and displacements. If we want to see the $x,y$ plane as the euclidean part of a projective plane corresponding to the factorization of a relevant vector-space of a mechanical problem, we should think of forces as triplets $\vec{f}=(-M_y,M_x,F_z)$ and displacements as triplets $\vec{d}=(-\Phi_y,\Phi_x,\Delta_z)$. This way point $(p_x,p_y,1)\sim F_z(p_x,p_y,1)=(-M_y,M_x,F_z)$ represents the point of attack of the force, while triplets $(-\Phi_y,\Phi_x,\Delta_z)$ represent lines around which the rotations happen.
In accordance with engineering practice, we assume the members have negligible rotational stiffness. The stiffness matrix of a  member with fixed supports on both ends (Figure \ref{fig:app_b1}, left) turns into:
\begin{align}
\vec{K}=\begin{bmatrix}
EI_y/L & & \\
 & 0 & \\
 &  & 12EI_y/L^3 
\end{bmatrix}
\end{align}
which needs to be interpreted as a dual conic, since we see forces as points and displacements as lines. The kernel of $\vec{K}$ is spanned by $(0,1,0)$, which represents the line $y=0$. This is nothing else than the axis of the rod, with a mechanical interpretation similar to the "primal" case. The only displacement not causing stresses is rotation around this axis (as we neglected rotational stiffness), and possible deformational forces from relative motion of the endpoints of the member must lie on this line.

\begin{figure}[h]
\centering
\includegraphics[width=1\textwidth]{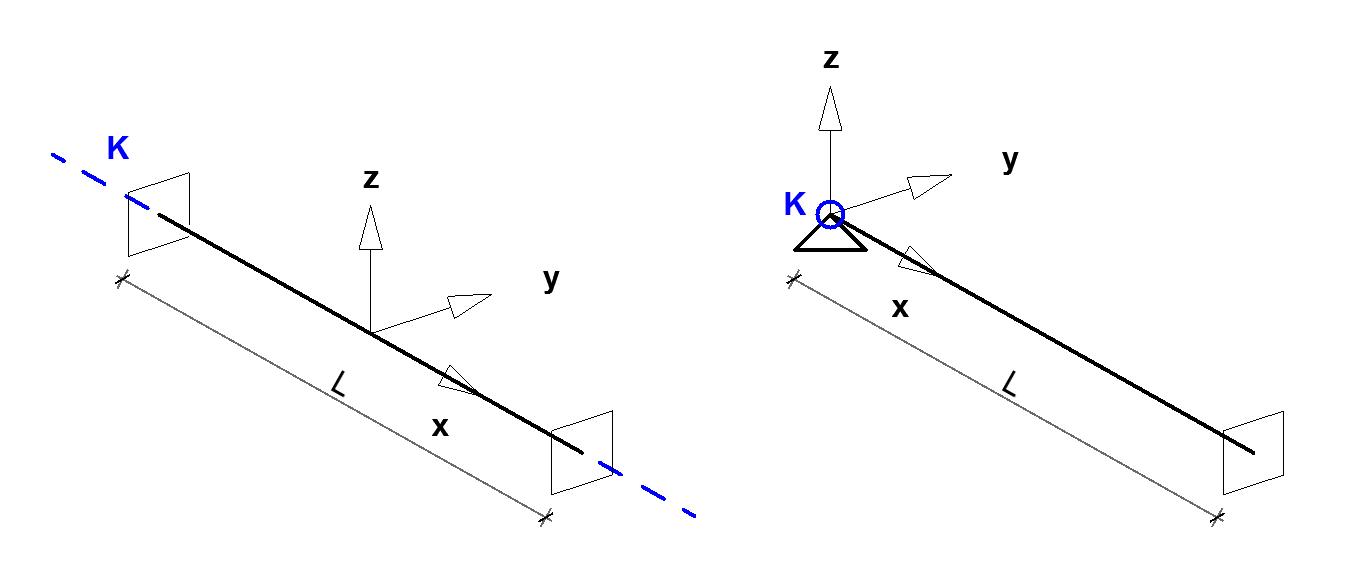}
\caption{Grillage members, with corresponding degenerate stiffness conics. The conics are points and lines in the $x,y$ plane, the $z$ direction is drawn for convenience}\label{fig:app_b1}
\end{figure}

A member supported with a fixed support on one end and a pinned support on the other in an appropriate coordinate system (Figure \ref{fig:app_b1}, right) has stiffness matrix
\begin{align}
\vec{K}=\begin{bmatrix}
0 & & \\
 & 0 & \\
 &  & 3EI_y/L^3 
\end{bmatrix}
\end{align}

The kernel of this is spanned by $\{(1,0,0),(0,1,0)\}$ and this 2 dimensional linear subspace corresponds to a one dimensional pencil of projective lines: all lines passing through the origin. Mechanically speaking the force from the displacement has to act at the pinned support and any rotation with axis passing through this point will not induce stresses and forces.

A simple grillage from these two elements is shown in Figure \ref{fig:app_b2}. The stiffness conic of vertex $i$ would be the sum of the stiffness conics of the two elements. As it is complex, the elasticity ellipse of vertex $i$ is shown. (One should express the conics of the members in a common frame before the addition.)

\begin{figure}[h!]
\centering
\includegraphics[width=1\textwidth]{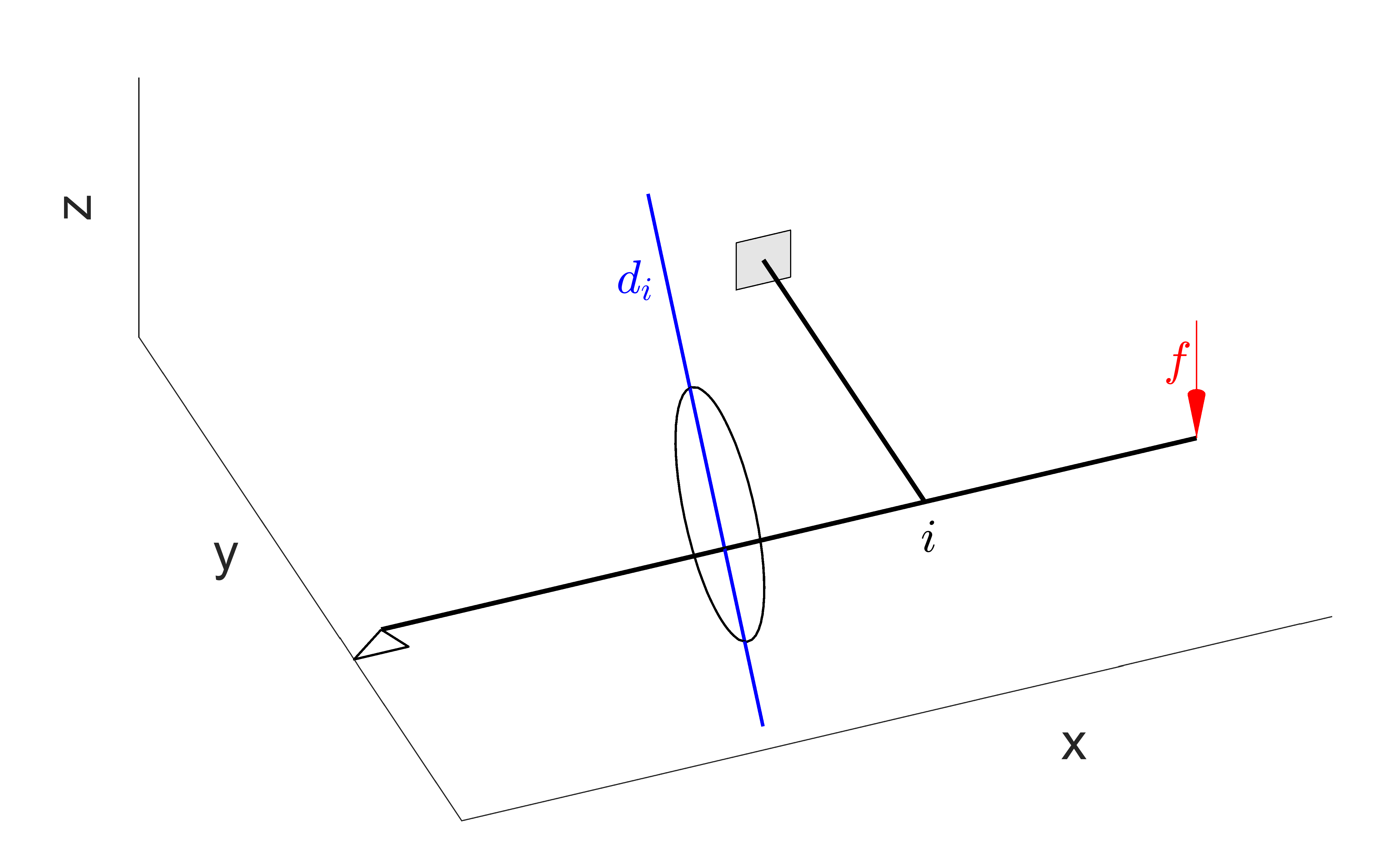}
\caption{Simple grillage, with the (dual) elasticity ellipse of vertex $i$. The point of attack of force $\vec{f}$ is the antipole of the line of rotation $\vec{d}_i$ (of vertex $i$) with respect to the elasticity ellipse}\label{fig:app_b2}
\end{figure}

\pagebreak

\bibliographystyle{unsrt} 
\bibliography{grafobib}

\end{document}